\documentclass[12pt]{iopart}

\usepackage{amsthm}
\usepackage{graphics}
\usepackage{graphicx}
\usepackage{subfigure}
\usepackage[linesnumbered,ruled,vlined]{algorithm2e}
\usepackage{amsfonts}
\usepackage{hyperref}

\newcounter{defcounter}
\newenvironment{definition}{
	\refstepcounter{defcounter}
	\par\noindent\textbf{Definition \thedefcounter.}\itshape
}{
	\par
}
\newcounter{procounter}
\newenvironment{proposition}{
	\refstepcounter{procounter}
	\par\noindent\textbf{Proposition \theprocounter.}\itshape
	\space
}{
	\par
}
\newcounter{thecounter}
\newenvironment{theorem}{
	\refstepcounter{thecounter}
	\par\noindent\textbf{Theorem \thethecounter.}\itshape
	\space
}{
	\par
}

\newcounter{exacounter}
\newenvironment{example}{
	\refstepcounter{exacounter}
	\par\noindent\textbf{Example \theexacounter.}\itshape
	\space
}{
	\par
}
\newcounter{lemcounter}
\newenvironment{lemma}{
	\refstepcounter{lemcounter}
	\par\noindent\textbf{Lemma \thelemcounter.}\itshape
	\space
}{
	\par
}

\newcounter{errcounter}
\newenvironment{errthm}{
	\refstepcounter{errcounter}
	\par\noindent\textbf{Erroneous Theorem \theerrcounter.}\itshape
	\space
}{
	\par
}

\newcounter{coucounter}
\newenvironment{counterexample}{
	\refstepcounter{coucounter}
	\par\noindent\textbf{Counterexample \theerrcounter.}\itshape
	\space
}{
	\par
}

\begin{document}

\title[The Invertibility of Cellular Automata with Menory]{The Invertibility of Cellular Automata with Menory: Correcting Errors and New Conclusions}

\author{Chen Wang$^1$ \footnote{2120220677@mail.nankai.edu.cn}, Xiang Deng$^2$ \footnote{2120230763@mail.nankai.edu.cn}, Chao Wang$^{3,*}$ \footnote{wangchao@nankai.edu.cn}}
\ead{2120220677@mail.nankai.edu.cn}
\ead{2120230763@mail.nankai.edu.cn}
\ead{wangchao@nankai.edu.cn}

\address{Address: College of Software, Nankai University, Tianjin 300350, China}

\vspace{10pt}

\begin{indented}
\item[]Jun 2024
\end{indented}

\begin{abstract}
	Cellular automata with memory (CAM) are widely used in fields such as image processing, pattern recognition, simulation, and cryptography. The invertibility of CAM is generally considered to be chaotic. Paper [Invertible behavior in elementary cellular automata with memory, Juan C. Seck-Tuoh-Mora et al., Information Sciences, 2012] presented necessary and sufficient conditions for the invertibility of elementary CAM, but it contains a critical error: it classifies identity CAM as non-invertible, whereas identity CAM is undoubtedly invertible. By integrating Amoroso's algorithm and cycle graphs, we provide the correct necessary and sufficient conditions for the invertibility of one-dimensional CAM. Additionally, we link CAM to a specific type of cellular automaton that is isomorphic to CAM, behaves identically, and has easily determinable invertibility. This makes it a promising alternative tool for CAM applications.
\end{abstract}
\vspace{2pc}

\section{Introduction \label{s1}}
\textbf{Cellular Automaton} (CA) is a discrete model, consisting of a rule and a grid of one or more dimensions. Each grid position has a state, and these states are updated in parallel at discrete time steps based on a fixed local rule \cite{1951_Neumann}. CA has a wide range of applications in simulation, encryption, pseudorandom number generation, and more \cite{jun2009image, kang2008pseudorandom, rosin2014cellular}.

Generally, CA are memoryless, meaning that the evolution of a CA at the next time step depends only on its current state and not on its historical states. \textbf{Cellular automata with memory} (CAM) is a variation of CA that incorporates historical memory into its functions. Due to the addition of memory, the behavior of CAM is more complex and multifunctional than that of CA, enabling it to perform a broader range of unconventional computations \cite{alonso2008cellular, alonso2007elementary, martinez2009complex, martinez2010make}. 

\textbf{Invertibility}, which relates to the preservation of information in discrete dynamic systems, has been widely studied from CA \cite{1979_Bruckner, MARTINDELREY20118360, 1991_Sutner, 1990_Kari} to CAM \cite{sanz2003reversible, alonso2007reversible, mcintosh1991reversible}. 
For a long time, the invertibility of CAM has been considered to be very complex: ``CA with memory in cells would result in a qualitatively different behavior" \cite{alonso2007elementary, wuensche1993global}. In 2012, Seck-Tuoh-Mora, etc. provided the necessary and sufficient conditions for the invertibility of CAM \cite{SECKTUOHMORA2012125}. However, we have found a counterexample to their theorem: they classified the identity CAM, which is invertible, as non-invertible. 

Since CAM do not belong to the strictly defined category of CA, existing algorithms for determining the invertibility of CA are incompatible, making the determination of CAM's invertibility more challenging. To address this, we introduce a new concept, \textbf{memory cellular automata} (MCA).
We prove that for any CAM, there exists an isomorphic MCA that exhibits the same behavior. Because MCA is a kind of CA, it has better properties and is compatible with existing algorithms developed for CA \cite{1972_Amoroso,wang2024computation}. It is anomalous that the invertibility of CAM is not equivalent to the invertibility of MCA. This discrepancy is one of the key reasons for the errors that have occurred in \cite{SECKTUOHMORA2012125}. Ultimately, we discovered that MCA is the key to addressing the complex invertibility of CAM. By integrating Amoroso's algorithm \cite{1972_Amoroso} and cycle graphs \cite{seck2018graphs}, we provided the correct necessary and sufficient conditions for the invertibility of CAM. Additionally, MCA exhibits the same evolutionary behavior as CAM but possesses superior properties and more extensive computational capabilities. This makes MCA a viable alternative model to CAM in many applications. The comparison of MCA and CAM is illustrated in Table \ref{table1}.

\begin{table}[h]
	\center
	\caption{Comparison of CAM and MCA}
	\label{table1}
	\begin{tabular}{p{5cm}p{4cm}p{4cm}}
		\hline
		& CAM & MCA\\
		\hline
		computational efficiency & high     & low    \\
		attribute                & complex  & explicit  \\
		algorithms               & less     & more    \\
		\hline
	\end{tabular}
\end{table}

This paper is comprised of five sections. The second section introduces the relevant definitions of CA and CAM, and discusses the erroneous theorem in \cite{SECKTUOHMORA2012125} along with its counterexample. In Section \ref{s3}, we introduce the isomorphic MCA of CAM, discuss the relationship between their invertibility, and identify the reason for the errors. In Section \ref{s4}, utilizing the cycle graph of MCA and Amoroso's algorithm, we successfully provide the necessary and sufficient conditions for the invertibility of CAM. The final section summarizes the entire work of this paper.

\section{Preliminaries and the erroneous theorem \label{s2}}
\begin{definition}
	A one-dimensional cellular automaton with periodic boundaries is defined by a triad: $A' = \{S, M, \varphi_m \}$
	\begin{itemize}
		\item $S=\{0,1,\ldots, p-1\}$ is a finite set of states representing a cell's state in CA. At any time, the state of each cell is an element in this set. We use $p$ to count the number of elements in this set.
		\item $M=(\vec c^1, \vec c^2, \ldots , \vec c^m)$ represents the neighbor vector, where $\vec c^i \in \mathbb{Z}^d$, and $\vec c^i \neq \vec c^j$ when $i \neq j$ $(i, j = 1, 2,\ldots, m)$. Thus, the neighbors of the cell $\vec c \in \mathbb{Z}^d$ are the $m$ cells $\vec c + \vec c^i, i = 1, 2,\cdots, m$. 
		\item $\varphi_m$: $S^m \rightarrow S$ is the rule function. The rule maps the current state of a cell and all its neighbors to this cell's next state. 
	\end{itemize}
\end{definition}

\begin{definition}
	Configuration is a set of snapshots of all automata in CA at a certain time, which can be represented by the function $C: \mathbb{Z} \rightarrow S$. These cells interact locally and transform their states in parallel. Function $\Phi$ is used to represent this global transformation. If there are two configurations $c_1$, $c_2 \in C$ and $\Phi(c_1) = c_2$, then $c_1$ is a predecessor of $c_2$ while $c_2$ is a successor of $c_1$. If $c_2$ has no predecessors, then it is a \textbf{Garden-of-Eden} configuration.
\end{definition}

CA are non-historical, meaning their evolution at any given moment depends only on their current configuration and is independent of their predecessors or ancestors. In contrast, the evolution of CAM is influenced not only by the current configuration but also by the past states of the cells. 
Therefore, the definition of one-dimensional CAM expands from the triad to a quadruple.

\begin{definition} \cite{SECKTUOHMORA2012125}
	A one-dimensional cellular automaton with memory (periodic boundary) is defined by a quadruple: $A = \{S, M, \varphi_m, \beta_n \}$
	\begin{itemize}
		\item The definitions of $S$, $M$ and $\varphi_m$ are the same as that in CA, representing the set of states, neighbor vector and rule. $\Phi$ is the global transformation function of $\varphi_m$.
		\item $\beta_n: S^n \rightarrow S$ is the memory function, where $n$ is the memory size. $B$ is the global transformation function corresponding to $\beta_n$. 
		In time $t$, $B$ acts on the current and past $n$ configurations $\Gamma_n(c_t) = (c_{t-n+1},\cdots,c_{t})$ to evolve an intermediate configuration $d_t$. Then $d_t$ is processed by $\Phi$ to generate the new configuration. 
	\end{itemize}	
\end{definition}

The dynamics of the CAM begin at a given initial configuration $c_1 \in C$, where the subscript indicates the current moment in time. The following equation shows its evolution:

\begin{equation}
	\label{CAM}
	c_{t+1}= \left\{ 
	\begin{array}{ll}
		\Phi(c_{t}) & \textrm{if } t<n, \\
		\Phi \circ B \circ \Gamma_n(c_{t}) & \textrm{if } t \ge n.
	\end{array}
	\right.
\end{equation}

\begin{itemize}
	\item When $t \ge n$, $\beta_n$ initially acts on the $n$ historical states of each cell $c^i_t$ in configuration $c_t$ to evolve an intermediate configuration $d_t$. Then, $d_t$ is processed by  $\Phi$ to obtain the new configuration $c_{t+1}$.
	\item When $t<n$, since the number of configurations stored has not reached the memory size, the function $B$ does not run and can be considered as an identity function.
\end{itemize}

\begin{figure}[h]
	\center
	\includegraphics[width=0.6\linewidth]{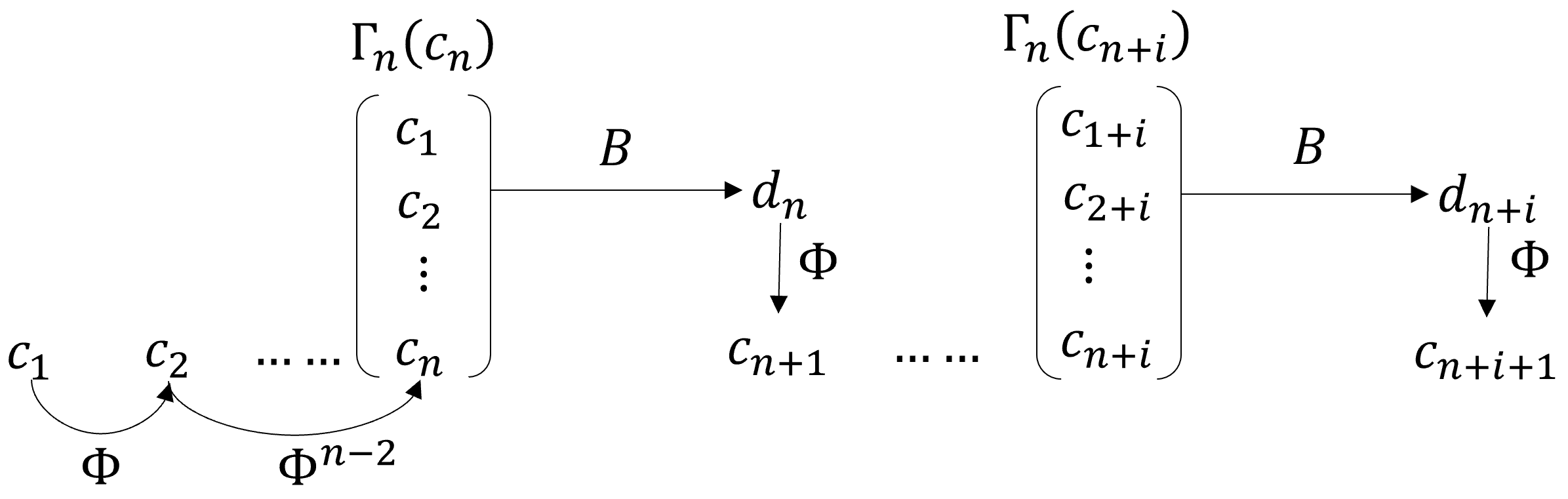}
	\caption{the evolution of CAM}
	\label{evolutionG}
\end{figure}

\begin{definition}
	A CAM is invertible iff for any $t>n$, $c_t$ is uniquely determined by $c_{t-n},\cdots,c_{t-1}$. This bijective mapping can be expressed as $c_{t-n} 
	\stackrel{\Gamma_{n-1}(c_{t-1})}{\longleftarrow} c_{t}$.
	
\end{definition}

Paper \cite{SECKTUOHMORA2012125} presents a theorem for the invertibility of CAM, but we hold that there is an error. Below is the theorem and the counterexample we have identified.

\begin{errthm}
	\label{err} 
	Suppose the CA: $A' = \{S, M, \varphi_m \}$ is invertible. The CAM: $A = \{S, M, \varphi_m, \beta_n \}$ is invertible iff $\beta_n$ is permutative, that is, $\forall x,y \in S$, $x \ne y$ and $\forall w \in S^{n-1}$, $\beta_n(xw) \ne \beta_n(yw)$ is satisfied.
\end{errthm}
\begin{counterexample}
	\label{CE}
	There is a counterexample that $\beta_n$ is not permutative but the CAM is invertible.
	
	Assuming $S=\{0,1\}$ and $\beta_n$, $\varphi_m$ are identity functions $(m=3, n=3)$, they are specifically given by the following equations:
	
	\begin{equation}
		\label{id}
		\beta_3(xyz)=z, \ \ \ \ \varphi_3(abc)=b.
	\end{equation}
	CA: $A' = \{S, M, \varphi_3 \}$ is identical, and it is easily shown to be invertible. Then the CAM: $A = \{S, M, \varphi_3, \beta_3 \}$ is also invertible, because for any initial configuration 
	$c_1$, all subsequent configurations will always be equal to $c_1$, meaning the inverse of every configuration is itself. The mapping table of the $\beta_3$ is shown in Eq. \ref{beta}, and it is easy to see that $\beta_3$ is not permutative.
	
	\begin{equation}
		\label{beta}
		\begin{array}{c|cccc}
			\beta_3(xyz)  & yz=00 & yz=01 & yz=10 &yz=11\\
			\hline
			x=0 & 0 & 1 & 0 & 1\\
			x=1 & 0 & 1 & 0 & 1\\
		\end{array}
	\end{equation}

	Therefore, there exists an invertible CAM where $\varphi_m$ is invertible but $\beta_n$ is not permutative, proving the theorem incorrect. Many similar CAM can be constructed, and we choose the most typical one as a counterexample since the invertibility of an identical CAM requires no explanation.
	
\end{counterexample}

\section{Isomorphic CA of CAM \label{s3}}
To address the issue mentioned in the previous section and correct the erroneous conclusion, we introduce a new concept: isomorphic CA for CAM. In this section, we will introduce the relationship between CAM and its isomorphic CA.

\begin{definition}	
	There is a bijective mapping $g:A \rightarrow A'$ where $A = \{S, M, \varphi_m, \beta_n \}$, $A'=\{S',M',f\}$ and
	\begin{itemize}
		\item $S'=S^n$. Each cell stores $n$ states which correspond to the memory size.
		\item $M'=M$. The neighborhood is the same.
		\item $f= \varphi_m \odot \beta_n$. Without losing generality, if $m$ is odd and the neighborhoods on both sides of a cell are equal, then it can be expressed as Eq. \ref{rule}. 
		\begin{eqnarray}
			\label{rule}
			\fl
			f(c_{t}^{i}c_{t+1}^{i} \cdots c_{t+n-1}^{i}, c_{t}^{i+1}c_{t+1}^{i+1} \cdots c_{t+n-1}^{i+1}, \cdots,
			c_{t}^{i+m-1}c_{t+1}^{i+m-1} \cdots c_{t+n-1}^{i+m-1}) \\
			\fl
			= c_{t+1}^{i+(m-1)/2}c_{t+2}^{i+(m-1)/2} \cdots c_{t+n-1}^{i+(m-1)/2}\varphi_m(\beta_n(c_{t}^{i}c_{t+1}^{i} \cdots c_{t+n-1}^{i}), \cdots, \beta_n(c_{t}^{i+m-1}c_{t+1}^{i+m-1} \cdots c_{t+n-1}^{i+m-1})) \nonumber
		\end{eqnarray}
		Apply the $\beta_n$ to the  $m$ cells within a cell's neighborhood. The resulting value is then processed by the $\varphi_m$ operation to obtain a new state value, which is set as the cell's latest state. The cell's original $n$ states are shifted forward one by one, and the earliest state is discarded. 
		The global transformation function corresponding to $f$ is denoted as $F$.
	\end{itemize}
	$A$ is isomorphic to $A'$ and we call $A'$ the memory CA (MCA).
\end{definition}

\begin{example}
	\label{e1}
	There is a mapping from the CAM introduced in Counterexample \ref{CE} to MCA: $A'=\{S',M',f\}$ where

	\begin{equation}
		\label{iso}
		\fl
		\left\{ 
		\begin{array}{ll}
			S' = S^3 \\
			M' = M\\
			f(x_1x_2x_3,y_1y_2y_3,z_1z_2z_3)=y_2y_3\varphi_3(\beta_3(x_1x_2x_3),\beta_3(y_1y_2y_3),\beta_3(z_1z_2z_3))=y_2y_3y_3
		\end{array}
		\right.
	\end{equation}
	
\end{example}

\begin{definition}
	For CAM A, the \textbf{initial configuration set}  $\hat{C_A}=\{(c_{1},\cdots,c_{n}) \in C_A^n\ | \forall i \in \mathbb{Z}, 1\le i<n, c_{i+1}=\Phi(c_{i})\}$, where $C_A$ is the configuration set of $A$. 
\end{definition}

\begin{definition}
	For CAM $A$, the \textbf{continuous configuration set} $\bar{C}_A = \{(c_{1+t},\cdots,c_{n+t}) \in C_A^n\ | t \in \mathbb{Z}_{+} \cup \{0\}, \exists (c_{1},\cdots,c_{n})\in \hat{C_A}, \forall i \in \mathbb{Z}_+, n-t <i \leq n, c_{i+t} = (\Phi \circ B \circ \Gamma_n)^{i+t-n}(c_{n}) \}$ is the set of continuous $n$ configurations during the evolution.
\end{definition}
\begin{theorem}
	For CAM $A= \{S, M, \varphi_m, \beta_n \}$, if $|S|=s$, the memory size is $n$ and the number of cells is $u$, then $\hat{C_A} \subseteq  \bar{C}_A$ and $s^u = |\hat{C_A}| \leq |\bar{C}_A| \le |C_{A}^{n}| = s^{un}$.
\end{theorem}
\begin{proof}
	It is easy to obtain $|\hat{C_A}| = s^u$ and $|C_{A}^{n}| = s^{un}$ because the number of elements in $\hat{C_A}$ is uniquely determined by the number of $c_1 \in C_{A}$. For each $\hat{c} \in \hat{C_A}$, there exist $\bar{c} \in \bar{C_A}$ and $\hat{c}=\bar{c}$ when $t=0$, hence $\hat{C_A} \subseteq \bar{C}_A$. The maximum value of $|\bar{C}_A|$ is $s^{un}$, but it often does not reach the maximum value. For example, in the CAM in Counterexample \ref{CE}, $|\hat{C_A}| \leq |\bar{C}_A| = s^u$.
\end{proof}

\begin{theorem}
	If $A \stackrel{g}{\rightarrow} A'$, the injective mapping $h_g:\bar{c} \rightarrow c$ maps the continuous configuration $\bar{c} \in \bar{C}_A$ to the configuration $c \in C_{A'}$ where the $n$ states of each cell in $c$ are equal to the $n$ historical states of the corresponding cell in $\bar{c}$ at the same position.
\end{theorem}
\begin{proof}
	We can uniquely map each element $\bar{c}$ to configurations of $A'$ where the state of each cell in $C_A$ consists of the $n$ states of the corresponding cell in $\bar{c}$. For $|\bar{C}_A| \le s^{un}, |C_{A'}| = s^{un}$, the map $h_g$ is injective.
\end{proof}

\begin{theorem}
	\label{relation}
	Assuming $A \stackrel{g}{\rightarrow} A'$, if $A'$ is invertible, then $A$ is invertible, but not vice versa. 
\end{theorem}
\begin{proof}
	Since $A'$ is invertible, $F$ can not cause many-to-one mapping of $C_{A'}$. Now, $h_g$ maps elements in $\bar{C}_A$ injectively that in $C_{A'}$, thus ensuring that the mapping $\Phi \circ B$ does not have many-to-one relationships. Therefore, $A$ is invertible.
	
	Conversely, if $A$ is invertible, although $\Phi \circ B$ cannot cause many-to-one relationships within $\bar{C}_A$, the number of elements in $\bar{C}_A$ is less than in $C_{A'}$. This only ensures that $F$ is invertible within the range of $h_g(\bar{C}_A)$ but does not guarantee that $F$ will not cause many-to-one relationships within the range of $C_{A'} \setminus h_g(\bar{C}_A)$. The error in Erroneous Theorem \ref{err} and Counterexample \ref{CE} lies in incorrectly equating the invertibility of $A$ and $A'$.
\end{proof}

\begin{theorem}
	MCA: $A'=\{S',M',f=\beta_n \odot \varphi_m\}$ is invertible iff $\varphi_m$ is invertible and $\beta_n$ is permutative.
\end{theorem}
\begin{proof}
	The invertibility of $A'$ is equivalent to the invertibility of its global mapping $F$. Each cell $A'$ has $n$ states, so each configuration $c \in C_{A'}$ can be regarded as a combination of $n$ single-state configurations $(c_i,\cdots,c_{i+n-1})\in C^n_{A}$. 
	The mapping relationship of the global transformation function $F$ is $(c_i,\cdots,c_{i+n-1}) \stackrel{F}{\rightarrow} (c_{i+1},\cdots,c_{i+n})$. And $f= \varphi_m \odot \beta_n$, so $F= \Phi \circ B$. Then the mapping $F$ can be decomposed into two parts, as shown in Fig. \ref{mapF}.
	
	\begin{figure}[h]
		\center
		\includegraphics[width=0.8\linewidth]{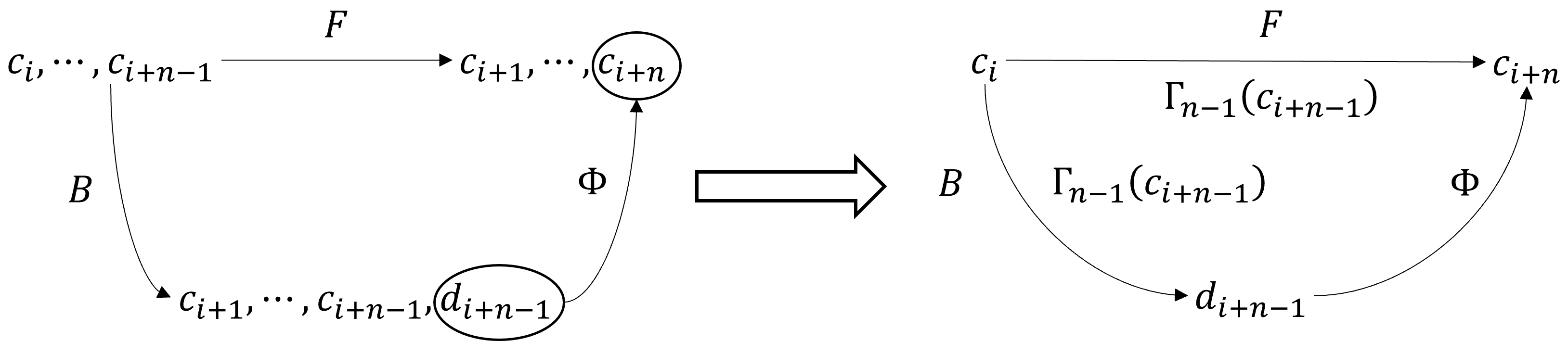}
		\caption{the invertibility of $MCA$}
		\label{mapF}
	\end{figure}

	On the surface, $F$ is the mapping on the left side of the arrow, but it can be simplified to the format on the right side. 
	Since $c_i,d_{i+n-1}$ and $c_{i+n}$ have the same number of values, both $B$ and $\Phi$ must form bijections for $F$ to be invertible. For $B$ to be bijective, $\beta_n$ must be permutative, that is $\forall x,y \in S$, $x \ne y$ and $\forall w \in S^{n-1}$, $\beta_n(xw) \ne \beta_n(yw)$ is satisfied.
	
\end{proof}

Compared to CAM, MCA has better properties because CAM does not strictly belong to the category of CA. This means that many conclusions and algorithms in CA cannot be applied. Comparatively, MCA is well-suited for Amoroso's surjectivity and injectivity algorithms, as well as other algorithms.

\section{The invertibility of CAM  \label{s4}}

We have successfully linked CAM with MCA bijectively and resolved the invertibility of MCA. However, the invertibility of CAM is a more complex issue. To address this, we need to introduce a type of graph widely used in CA.

\begin{definition}
	\cite{seck2018graphs} For CA $A'$, the cycle graph $G=(V,E)$ is defined as follows:
	\begin{itemize}
		\item The vertex set $V$ is equal to the configuration set $C_{A'}$.
		\item There exists a directed edge $(v_1,v_2) \in E$ from vertex $v_1$ to $v_2$ iff $\Phi(v_1)=v_2$.
	\end{itemize}
	
\end{definition}

\begin{proposition}
	\cite{seck2018graphs} A CA is invertible iff its cycle graph contains no branches.
\end{proposition}

The invertibility of CA is directly related to its cycle graph, but cycle graphs have not been widely applied to determinations of CA invertibility. The primary reason is that the scale of cycle graphs is excessively large, reaching exponential levels relative to the number of cells. Constructing cycle graphs sequentially for a CA, starting from one cell to a large number, is nearly an impossible task. However, cycle graphs can be of great assistance in solving the invertibility issues of CAM, because we only need to construct a limited number of connected components of a cycle graph for a CA.

\begin{lemma}
	\label{branch}
	CAM $A$ is invertible iff all on-branch configurations in the cycle graph of its isomorphic MCA $g(A)$ are not in $h_g(\bar{C_A})$.
\end{lemma}
\begin{proof}
	It has been proven in Theorem \ref{relation} that if there are no Garden-of-Eden configurations in $g(A)$, then $A$ is invertible. Now, suppose that there exist Garden-of-Eden configurations in $g(A)$; in this case, there must be branches on its cycle graph, and their preimage of $h_g$ must also be Garden-of-Eden configurations of $A$. If all configurations on branches are not in $h_g(\bar{C_A})$, then the cycle graph of $A$ also has no branches, implying that $A$ is invertible. Conversely, if there exist configurations on the branches of $g(A)$ in $h_g(\bar{C_A})$, then the cycle graph of $A$ will have branches.
\end{proof}

\begin{lemma}
	\label{ancestor}
	A configuration $c \in h_g(\bar{C}_A)$ iff it have an ancestor (including itself) $c' \in h_g(\hat{C}_A)$.
\end{lemma}
\begin{proof}
	As shown in Eq. \ref{CAM}, when $t \leq n$, only the function $\Phi$ operates in the evolution of CAM, and $\beta_n$ or $B$ can be considered as an identity function. The state at $t=n$, denoted as $(c_{1},\cdots,c_{n}) \in \hat{C_A}$, definitely satisfies $\forall i \in \mathbb{Z}, 1\le i<n, c_{i+1}=\Phi(c_{i})$. These configurations are called initial configurations, from which all configurations of CAM evolve. Therefore, if a configuration does not have an ancestor in $\hat{C_A}$, the time $n$ does not exist, so it is not included in the continuous configuration set $\bar{C_A}$.
\end{proof}

\begin{theorem}
	\label{fin}
	CAM $A$ is invertible iff all on-branch configurations in the cycle graph of its isomorphic MCA $g(A)$ are not in $h_g(\hat{C_A})$.
\end{theorem}
\begin{proof}
	Since every node in the cycle graph has an out-degree of one, if branches exist, they must point toward a circuit. If there is no configuration $c \in h_g{(\hat{C_A})}$ on the branch, then the entire branch will not appear in the evolution of $A$, making $A$ invertible. Conversely, if there is a configuration $c \in h_g{(\hat{C_A})}$ on the branch the intersection vertex of the branch and circuit will have at least two predecessors—one from the branch and another from the circuit—making $A$ non-invertible.
\end{proof}

Theorem \ref{fin} is efficient because we can quickly determine whether a configuration is in $h_g(\hat{C_A})$ or not. Based on Theorem \ref{fin}, we can now fully determine the invertibility of CAM. The specific algorithm is as follows:
\begin{description}
	\item[Step 1] For CAM $A$, identify its isomorphic MCA $g(A)$.
	\item[Step 2] Find all Garden-of-Eden configurations of $g(A)$ according to Algorithm \ref{a1}, which modified from \cite{1972_Amoroso} to accommodate periodic boundaries.
	\item[Step 3] For each Garden-of-Eden configuration, construct its weakly connected component of the cycle graph.
	\item[Step 4] Determine whether there are initial configurations on the branches. If there are, then $A$ is not invertible. If there are none, then $A$ is invertible.
\end{description}

\begin{algorithm}[h]
	\small
	\SetAlgoLined
	\label{a1}
	\KwData{MCA $\{S, M, f\}$}
	\KwResult{Garden-of-Eden configuration set $P$}
	Let $G$: a directed graph; $L$: an empty queue; $P$: an empty set\;	
	$N_{ini} \leftarrow \{a_1a_2\cdots a_{m-1}a_1a_2\cdots a_{m-1}|a_1a_2\cdots a_{m-1} \in S^{m-1}\}$\;
	add $N_{ini}$ to $G$ and $L$\; 
	\While{$L$ is not empty}{
		pop $L$ and denote it with $N_{cur}$\;
		\If{$N_{cur} \cap N_{ini} == \emptyset$}{
			The path $p$ from $N_{cur}$ to $N_{ini}$ is a Garden-of-Eden configuration.\;
			add $p$ to $P$\;
		}
		\If{$not \exists N' = N_{child(c)}$ in $G$}{
			\For{each $c \in S$}{	
				$N_{child(c)} \leftarrow \{a_1a_2\cdots a_{m-1}b_1b_2\cdots b_{m-1}|a_1a_2\cdots a_{m-1}b_1b_2\cdots b_{m-1} \in S^{2m-2}, \exists b_0 \in S, a_1a_2\cdots a_{m-1}b_0b_1\cdots b_{m-2} \in N_{cur}$ and $f(b_0b_1\cdots b_{m-1}) = c \}$\;
				add an edge labeled $c$ from $N_{cur}$ to $N_{child(c)}$\;
				add $N_{child(c)}$ to $G$ and $L$\;
			}
		}
		
	}
	return $P$\;
	\caption{Garden-of-Eden configurations}
\end{algorithm}

\begin{example}
	We reassess the invertibility of the CAM in Counterexample 1. This CAM is invertible.
	\begin{description}
		\item[Step 1] Its isomorphic MCA $g(A)$ has been obtained in Example \ref{e1} and Eq. \ref{iso}.
		\item[Step 2] The Garden-of-Eden configurations of $g(A)$ include four configurations with one cell: 001, 010, 101 and 110.
		\item[Step 3] For each Garden-of-Eden configuration, construct the cycle graph as shown in Fig. \ref{CG}.
		\begin{figure}[h]
			\center
			\includegraphics[width=0.4\linewidth]{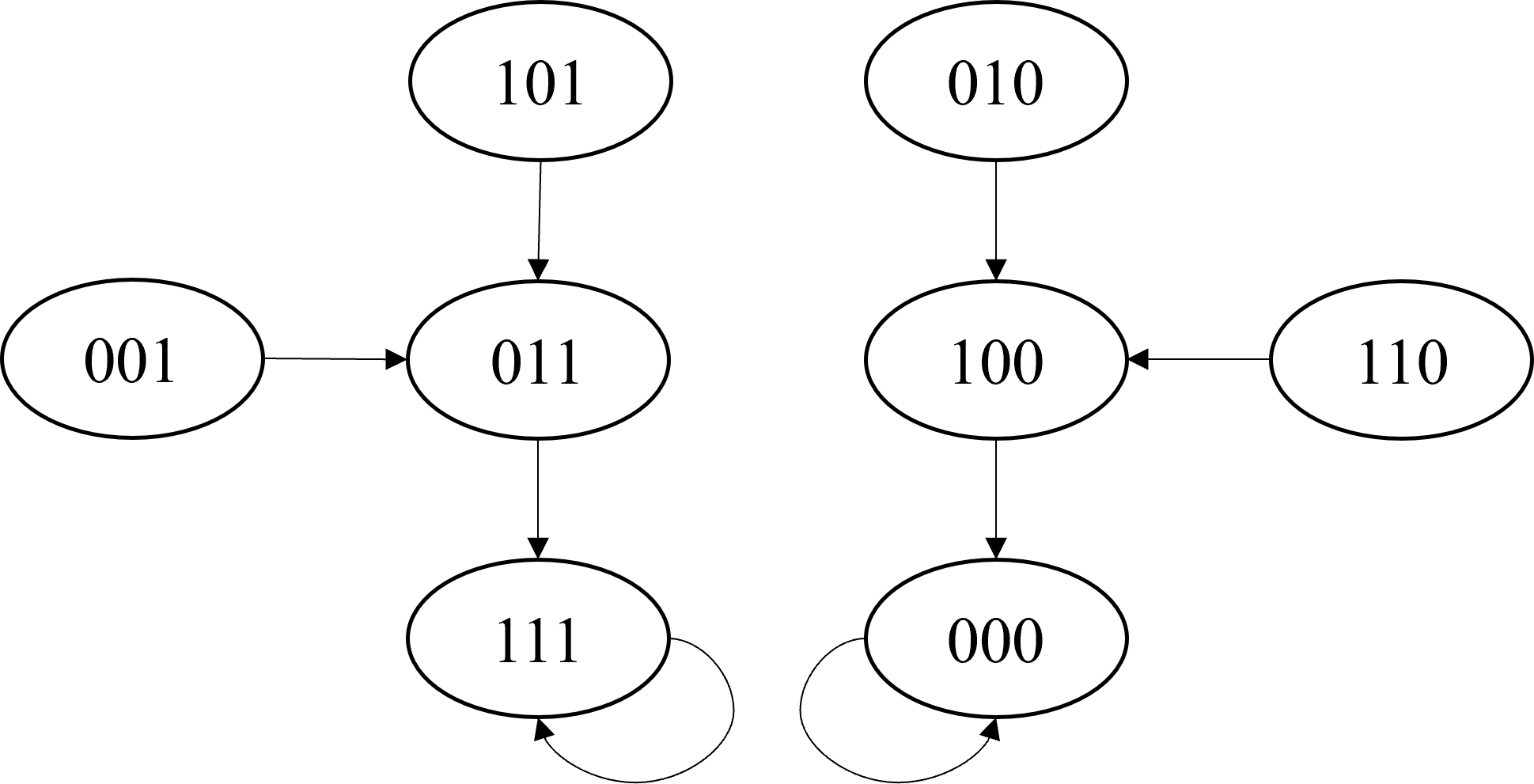}
			\caption{the cycle graph of $g(A)$}
			\label{CG}
		\end{figure}
		\item[Step 4] The six nodes on the branches in Fig. \ref{CG} are not within the initial configurations (requiring $a=b=c$), thus this CAM is invertible.
	\end{description}
\end{example}

\section{Conclusion  \label{s5}}

In this paper, we identified errors in the existing conclusions regarding the invertibility of CAM, introduced the concept of MCA, and resolved the invertibility issues of CAM. Both MCA and CAM, as discrete dynamic systems with memory, have their detailed relationships of invertibility illustrated in Fig. \ref{MCAM} We not only identified errors that the incorrect theorem could cause in practical applications but also provided a more suitable model for certain scenarios, that is the MCA.
\begin{figure}[h]
	\center
	\includegraphics[width=0.45\linewidth]{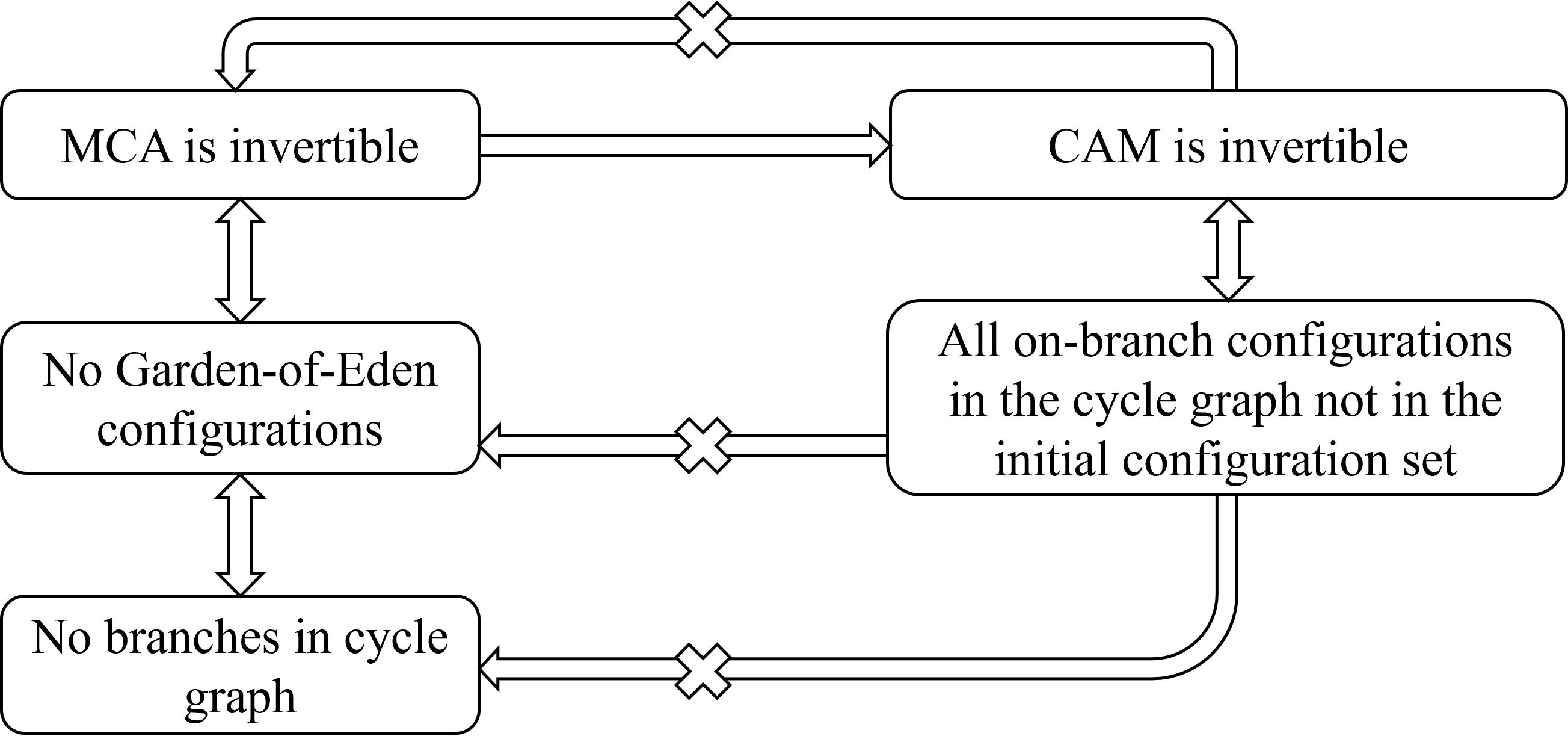}
	\caption{the invertibility relationship between CAM and MCA}
	\label{MCAM}
\end{figure}
%

%
%
%
%

\section*{Acknowledgments}
This study is financed by Tianjin Science and Technology Bureau, finance code:21JCYBJC00210.

\bibliographystyle{unsrt}
\bibliography{references}

\end{document}